\documentclass[preprint]{article}
\usepackage{staix_2026}

\usepackage[utf8]{inputenc}
\usepackage[T1]{fontenc}
\usepackage{hyperref}
\usepackage{url}
\usepackage{booktabs}
\usepackage{amsfonts,amsmath, amsthm}
\usepackage{nicefrac}
\usepackage{microtype}
\usepackage{xcolor}
\usepackage{graphicx}
\usepackage{xspace} 
\newcommand{\Wahkon}{\textsc{Wahkon}\xspace}

\title{Wahkon: A Statistically Principled Deep
RKHS Superposition Network}
\author{%
  First Author \\
  Affiliation \\
  Address \\
  \texttt{first.author@example.edu}
  \And
  Second Author \\
  Affiliation \\
  Address \\
  \texttt{second.author@example.edu}
}

 \author{
 Yongkai Chen\\
     Department of Statistics, Harvard University\\
     \texttt{yongkaichen@fas.harvard.edu}
     \And
     Wenxuan Zhong\\
     Department of Statistics,  University of Georgia\\
     \texttt{wenxuan@uga.edu}
     \And
     Ping Ma\\
     Department of Statistics,  University of Georgia\\
     \texttt{pingma@uga.edu}
 }
\date{} 

\newtheorem{theorem}{Theorem}[section]

\newtheorem{remark}{Remark}

\newtheorem{proposition}[theorem]{Proposition}
\newtheorem{assumption}[theorem]{Assumption}
 
\newcommand{\bfm}[1]{\ensuremath{\mathbf{#1}}}

     \def\EE{\mathbb{E}}

\def\bx{\bfm x}

 \def\cH{{\cal  H}}

 \def\cK{{\cal  K}}

\def\Var{\mathrm{Var}}

\def\a\cos{\mathrm{arc\cos}}

\begin{document}

\maketitle

\begin{abstract}
Deep learning excels at prediction but often lacks finite-sample guarantees and calibrated uncertainty; RKHS (Reproducing Kernel Hilbert Space)-based methods provide those guarantees but struggle to adapt in high dimensions. We propose Wahkon, a deep RKHS superposition network that unifies
Kolmogorov's superposition principle with RKHS regularization in the
smoothing-spline tradition of Wahba.
This yields a finite-dimensional \emph{deep representer theorem} that makes training tractable and provides explicit layerwise complexity control. We show the penalized estimator is exactly the MAP (maximum a posteriori) estimate under a hierarchical Gaussian-process prior, extending the spline/GP duality to deep compositions.
Using metric-entropy arguments, we establish minimax-optimal convergence rates under mild smoothness and clarify how depth and width trade off with regularity.
Empirically, Wahkon outperforms multilayer perceptrons, Neural Tangent Kernels, and Kolmogorov--Arnold Networks across simulation benchmarks and a single-cell CITE-seq study.
By unifying Kolmogorov's superposition principle with RKHS regularization, Wahkon delivers accuracy, interpretability, and statistical rigor in a single framework.
\end{abstract}

\section{Introduction}\label{sec:introduction}

Modern neural networks have transformed prediction across science and engineering \citep{Goodfellow-et-al-2016deep}, but their use for formal statistical inference remains challenging. Standard architectures typically rely on predetermined activation functions, such as ReLU or sigmoid, together with regularization and stabilization devices such as weight decay, dropout, normalization, and early stopping. While these tools are highly effective for prediction, they do not by themselves provide the function-space complexity control, uncertainty quantification, or interpretability that are central to many statistical applications. In contrast, classical nonparametric methods grounded in reproducing kernel Hilbert spaces (RKHS) offer explicit smoothness control, optimal convergence rates, and a close Bayesian interpretation through Gaussian processes (GPs) \citep{wahba1990spline,gu:13}. Neural tangent kernel (NTK) theory extends this kernel perspective to overparameterized neural networks by showing that, in suitable infinite-width and lazy-training regimes, gradient descent behaves like kernel regression with a deterministic kernel determined by the architecture and initialization \citep{jacot2018neural,lee2019wide}. However, classical kernel methods and NTK-based methods typically rely on kernels that are fixed before training or only weakly adaptive. Such kernels can impose a geometry and smoothness structure that may be poorly aligned with high-dimensional, heterogeneous data. Bridging these paradigms calls for an approach that preserves deep learning's adaptive representation power while restoring RKHS-style complexity control and principled statistical inference.

We introduce Wahkon, a method inspired by Kolmogorov’s superposition theorem \citep{kolmogorov1957representation,arnold1963representation} and the recent Kolmogorov–Arnold network (KAN) paradigm, which combines deep neural architectures with learnable activation functions. 
Wahkon adopts a standard deep compositional structure but replaces fixed activations with \emph{learned univariate} link functions along every edge of the network, each constrained to lie in an RKHS in the sense of Wahba’s smoothing-spline framework. 
The deep compositional structure enables Wahkon with high representational power and the RKHS-based regularization at the level of individual links provides a hierarchical smoothness control while remaining interpretable.

The article makes three main contributions. First, we establish a \emph{deep representer theorem} for Wahkon, despite the infinite-dimensional function space for each link function and the deep compositional architecture, the estimator admits a finite kernel expansion at every layer. This extends KAN-style architectures to an RKHS setting, replacing fixed B-spline bases and heuristic penalties with principled RKHS regularization, which in turn yields improved statistical convergence guarantees.
Second, we show that the penalized Wahkon estimator coincides with the \emph{Maximum a posteriori} with a deep hierarchical Gaussian-process prior on the link functions. This extends the classical spline/GP duality \citep{wahba1978improper,rasmussen2006gaussian} to the deep network structure, which provides insights into the success of Wahkon.
Third, we derive \emph{minimax-optimal} convergence rates under mild smoothness assumptions and make explicit how depth, width, and smoothness of link functions interact, providing principled guidance for hyperparameter selection. 
Empirically, Wahkon outperforms multilayer perceptrons (MLPs), NTK, and KANs in synthetic benchmarks. 
 In a CITE-seq application, Wahkon accurately predicts surface protein abundance from RNA \citep{stoeckius2017large,Stuart2019}.
Collectively, these results demonstrate that Wahkon combines the representational power of deep architectures with the rigor and stability of RKHS-based regularization, yielding a flexible yet statistically principled framework for high-dimensional prediction and inference.

\vspace{-8 pt}
\section{Background and Motivation}\label{sec:background}


\paragraph{Kolmogorov's Superposition Principle}

\begin{theorem}[\cite{kolmogorov1957representation}]
For any continuous $f:[0,1]^D\!\to\!\mathbb{R}$, there exist continuous univariate functions $\{\phi_{qk}\}$ and $\{\psi_q\}$ such that
$f(\bx)=\sum_{q=1}^{2D+1}\psi_q\!\left(\sum_{k=1}^D \phi_{qk}(x_k)\right)$ .
\end{theorem}

While nonconstructive, this theorem implies that any continuous multivariate function can be approximated by compositions of univariate functions. However, practical implementations prioritize computational efficiency over exact representation, diverging from the theorem's idealized structure.

\paragraph{Multilayer Perceptron Networks}

Multilayer perceptrons (MLPs) approximate a target by composing affine maps with a fixed nonlinearity, typically ReLU \citep{Goodfellow-et-al-2016deep}. They are universal approximators \citep[e.g.,][]{barron1993universal} and highly effective in practice, but fixed activations limit adaptivity and formal uncertainty.

\paragraph{Kolmogorov--Arnold Networks}

Kolmogorov--Arnold Networks (KANs) \citep{liu2024kan} replace fixed activations with learned univariate functions and summation across inputs. This brings depth efficiency and interpretability but lacks explicit statistical regularization and finite-sample guarantees.



\paragraph{Other related works}
Spline ANOVA and RKHS methods provide smoothness control, representer structure, and Bayesian interpretation \citep{wahba1990spline,gu:13,wahba1978improper}. Deep kernel learning \citep {wilson2016deep} bring kernels to multilayer neural network settings.
Neural Tangent Kernel \citep{jacot2018neural} and its infinite-width GP correspondence \citep{lee2018deep,lee2019wide} connect deep networks to GPs but use a lazy-training regime without feature learning.
A growing body of theoretical work studies unfixed link functions.
 \citet{horowitz2007rate} establishes rate-optimal estimation for generalized additive models with unknown link functions. 
  For fully connected deep networks, \citet{schmidt2020nonparametric} derive minimax-optimal rates over H\"older-smooth link functions and show that depth enables adaptation to compositional structure. 
  \citet{bauer2019deep} proves that deep networks overcome the curse of dimensionality when the it admits a compositional representation.
 

\section{The Wahkon Network}\label{sec:method}


\subsection{Network Architecture for Multivariate Regression}
\label{subsec:wahkon_architecture}

Consider the nonparametric regression model
$Y = \eta(\mathbf{X})+\epsilon$, $\mathbf{X}\in\mathbb{R}^D$, $\mathbb{E}(\epsilon\mid \mathbf{X})=0$, so that $\eta(\mathbf{x})=\mathbb{E}(Y\mid \mathbf{X}=\mathbf{x})$ is the regression function of interest. 
A Wahkon network approximates $\eta$ by a composition of $L$ layers. Its architecture is specified by the sequence of layer widths $(D_0,D_1,\ldots,D_L)$, where $D_0=D$ is the input dimension and $D_L=1$ for scalar-response regression. For a given input $\mathbf{x}\in\mathbb{R}^D$, define the layer outputs recursively. Set
$\mathbf{x}^{(0)}=\mathbf{x}$, 
$\mathbf{x}^{(l)}=(x^{(l)}_1,\ldots,x^{(l)}_{D_l})^\top\in\mathbb{R}^{D_l}$, $l=1,\ldots,L$. The $l$th layer is a map
$\Phi^{(l)}:\mathbb{R}^{D_{l-1}}\to \mathbb{R}^{D_l}$,
$\mathbf{x}^{(l)}=\Phi^{(l)}(\mathbf{x}^{(l-1)})$. 
Each coordinate of $\Phi^{(l)}$ is formed by summing univariate transformations of the coordinates from the previous layer:
\begin{equation}\label{eq:wahkon-layer}
x^{(l)}_{j}
=
\sum_{k=1}^{D_{l-1}}
\phi^{(l)}_{j k}\!\left(x^{(l-1)}_{k}\right),
\qquad 
j=1,\ldots,D_l,\quad l=1,\ldots,L.
\end{equation}
Here $\phi^{(l)}_{jk}:\mathbb{R}\to\mathbb{R}$ is the univariate link function from the $k$th unit in layer $l-1$ to the $j$th unit in layer $l$. Unlike conventional neural networks, where nonlinearities are usually fixed in advance, Wahkon learns these link functions from data. We assume each link belongs to a common reproducing kernel Hilbert space $\mathcal{H}_{\mathcal{K}}$ with reproducing kernel $\mathcal{K}$. The resulting network predictor is the scalar output of the final layer:
\begin{equation}\label{eq:wahkon-arch}
\eta_{\mathrm{Wahkon}}(\mathbf{x})
=
x^{(L)}_1
=
\Phi^{(L)}\circ \Phi^{(L-1)}\circ\cdots\circ \Phi^{(1)}(\mathbf{x}).
\end{equation}



\subsection{Penalized Estimation}
\label{subsec:penalized_estimation}

Let $\{(\mathbf{x}_i,y_i)\}_{i=1}^n$ be independent observations with
$\mathbf{x}_i\in\mathbb{R}^D$ and $y_i\in\mathbb{R}$. For a given collection of
link functions, let
$\mathbf{x}_i^{(l)} =(x^{(l)}_{i1},\ldots,x^{(l)}_{iD_l})^\top \in\mathbb{R}^{D_l}$,
 $l=0,\ldots,L$, denote the layer-$l$ output for the $i$th observation with
$\mathbf{x}_i^{(0)}=\mathbf{x}_i$. The fitted value for observation $i$ is
$x^{(L)}_{i1}=\eta_{\text{Wahkon}}(\mathbf{x}_i)$. We estimate the link functions by minimizing a penalized least-squares,
\begin{equation}\label{eq:loss}
\mathcal{L}(\eta_{\text{Wahkon}})
=
\sum_{i=1}^{n}
\bigl[y_i-\eta_{\text{Wahkon}}(\mathbf{x}_i)\bigr]^2
+
n\,J(\eta_{\text{Wahkon}}),
\end{equation}
where the penalty is
\begin{equation}\label{eq:penalty}
J(\eta_{\text{Wahkon}})
=
\sum_{l=1}^{L}\lambda_l
\sum_{j=1}^{D_l}\sum_{k=1}^{D_{l-1}}
\bigl\|\phi^{(l)}_{jk}\bigr\|_{\mathcal{H}_{\mathcal K}}^2,
\qquad \lambda_l>0.
\end{equation}
Here $\lambda_l$ is the regularization parameter for layer $l$. Larger
$\lambda_l$ enforces smoother and smaller link functions, while a smaller $\lambda_l$ allows more flexible nonlinear transformations. Unlike
weight decay in a standard MLP, this penalty is imposed  on the  functions
$\phi^{(l)}_{jk}$ through their RKHS norms.

\subsection{Deep Representer Theorem}
\label{subsec:representer}

The optimization problem in \eqref{eq:loss} is infinite-dimensional because each
link function $\phi^{(l)}_{jk}$ ranges over an RKHS. The following result shows
that the minimizer nevertheless has a finite representation. The kernel centers
at layer $l$ are the fitted inputs entering that layer, namely
$\{x^{(l-1)}_{ik}:i=1,\ldots,n\}$.

\begin{theorem}[Deep representer theorem]\label{thm:deep-representer}
Suppose each link function $\phi^{(l)}_{jk}$ belongs to
$\mathcal{H}_{\mathcal K}$ and that a minimizer of \eqref{eq:loss} exists. Then,
at any minimizer,
\[
\phi^{(l)}_{jk}(t)
=
\sum_{i=1}^{n}
c^{(l)}_{ijk}\,
\mathcal K\!\left(x^{(l-1)}_{ik},t\right),
\qquad
j=1,\ldots,D_l,\quad k=1,\ldots,D_{l-1},\quad l=1,\ldots,L,
\]
for some coefficients $c^{(l)}_{ijk}\in\mathbb{R}$. Therefore, the fitted Wahkon network is fully determined by the finite collection of coefficient
tensors $\mathcal C = \{\mathbf C^{(l)}:l=1,\ldots,L\}$,  
$\mathbf C^{(l)}\in\mathbb{R}^{n\times D_l\times D_{l-1}}$, whose $(i,j,k)$th entry is $c^{(l)}_{ijk}$.
\end{theorem}

For each layer $l$ and input coordinate $k$ of that layer, define the kernel
matrix
$\mathbf Q^{(l-1)}_{k}
=
\left[
\mathcal K\!\left(x^{(l-1)}_{ik},x^{(l-1)}_{i'k}\right)
\right]_{i,i'=1}^{n}$. For $l=1$, this matrix is computed from the observed predictors. For deeper
layers, it is computed from the fitted outputs of the previous layer. Let $\mathbf c^{(l)}_{\cdot jk}
=
(c^{(l)}_{1jk},\ldots,c^{(l)}_{njk})^\top$ be the coefficient vector for the link from coordinate $k$ in layer $l-1$ to
coordinate $j$ in layer $l$. By the reproducing property,
$\bigl\|\phi^{(l)}_{jk}\bigr\|_{\mathcal{H}_{\mathcal K}}^2
=
\mathbf c^{(l)\top}_{\cdot jk}
\mathbf Q^{(l-1)}_{k}
\mathbf c^{(l)}_{\cdot jk}
\equiv
\bigl\|\mathbf c^{(l)}_{\cdot jk}\bigr\|_{\mathbf Q^{(l-1)}_{k}}^2$. Consequently, the infinite-dimensional problem in \eqref{eq:loss} reduces to an
optimization over the finite coefficient collection $\mathcal C$:
\begin{equation}\label{eq:pls_C}
\mathcal L(\mathcal C)
=
\bigl\|\mathbf y-\mathbf X^{(L)}(\mathcal C)\bigr\|_2^2
+
n\sum_{l=1}^{L}\lambda_l
\sum_{j=1}^{D_l}\sum_{k=1}^{D_{l-1}}
\bigl\|\mathbf c^{(l)}_{\cdot jk}\bigr\|_{\mathbf Q^{(l-1)}_{k}}^2,
\end{equation}
where $\mathbf y=(y_1,\ldots,y_n)^\top$ and
$\mathbf X^{(L)}(\mathcal C)
=
\bigl(x^{(L)}_{11}(\mathcal C),\ldots,x^{(L)}_{n1}(\mathcal C)\bigr)^\top$ is the vector of final-layer fitted values. The penalty  in
\eqref{eq:pls_C} is quadratic in the coefficients for fixed layer inputs, while
the residual is generally nonlinear because the layer outputs are obtained
recursively through the forward map.

\subsection{Profile Objective}
\label{sec:profile}

The finite-dimensional objective in \eqref{eq:pls_C} can be simplified further by
profiling out the final layer. This is useful because, conditional on the outputs
of the first \(L-1\) layers, the last layer is an ordinary kernel ridge regression
problem. Let $\mathcal C_{<L}=\{\mathbf C^{(1)},\ldots,\mathbf C^{(L-1)}\}$ denote the lower-layer coefficients. For fixed \(\mathcal C_{<L}\), the layer-\((L-1)\)
outputs \(x^{(L-1)}_{ik}\), \(i=1,\ldots,n\), \(k=1,\ldots,D_{L-1}\), are fixed.
Define 
the aggregate final-layer kernel matrix
$
\mathbf K^{(L-1)}
=
\sum_{k=1}^{D_{L-1}}\mathbf Q^{(L-1)}_{k}.
$
For fixed \(\mathcal C_{<L}\), minimizing \eqref{eq:pls_C} over the final-layer
coefficients lead to the profile objective
\begin{equation}\label{eq:profile}
\mathcal P(\mathcal C_{<L})
= 
\min_{\{\mathbf c^{(L)}_{\cdot 1k}\}_{k=1}^{D_{L-1}}} \mathcal L(\mathcal C) = 
n\lambda_L\,
\mathbf y^\top
\bigl(\mathbf K^{(L-1)}+n\lambda_L\mathbf I_n\bigr)^{-1}
\mathbf y
+
n\sum_{l=1}^{L-1}\lambda_l
\sum_{j=1}^{D_l}\sum_{k=1}^{D_{l-1}}
\bigl\|\mathbf c^{(l)}_{\cdot jk}\bigr\|^2_{\mathbf Q^{(l-1)}_k}.
\end{equation}
The derivation is provided in the Appendix, Section \ref{sec:derive_Profile}.
The first term is the residual-plus-penalty value after optimally fitting the
last layer; the second term retains the RKHS penalties for the lower-layer link
functions. Minimizing \(\mathcal P(\mathcal C_{<L})\) over the lower-layer
coefficients is equivalent to minimizing the original objective
\eqref{eq:pls_C} over all coefficients. When the kernel is differentiable and
\(\mathbf K^{(L-1)}+n\lambda_L\mathbf I_n\) is nonsingular, the envelope theorem
justifies differentiating \(\mathcal P\) with respect to \(\mathcal C_{<L}\)
without differentiating through the profiled final-layer coefficients.

This profiling step improves optimization in practice (Section \ref{app:profile-empirical}, Appendix) because gradient-based
updates focus on learning the lower-layer representation, while the final layer
is solved exactly at each iteration. The idea is related to deep kernel learning
\citep{wilson2016deep}, which also learns a representation while solving a
kernel model at the output layer. The objectives differ, however: deep kernel
learning typically maximizes a Gaussian-process marginal likelihood, whereas
Wahkon directly minimizes a penalized least-squares criterion with explicit
RKHS penalties on the learned link functions for all layers.

\section{Wahkon as a Hierarchical Bayesian Model}
\label{sec:bayes}

The RKHS penalty in \eqref{eq:penalty} has a natural Bayesian interpretation. 
In classical smoothing-spline theory, penalized least squares can be viewed as a maximum a posteriori (MAP) estimator under a Gaussian-process prior \citep{wahba1978improper,wahba1990spline}. 
Wahkon extends this correspondence to a deep architecture by assigning Gaussian-process priors to the univariate link functions at each layer.

\subsection{Prior specification}

Recall that the link function from coordinate \(k\) in layer \(l-1\) to coordinate \(j\) in layer \(l\) is denoted by $\phi^{(l)}_{jk}:\mathbb{R}\to\mathbb{R}$, 
$j=1,\ldots,D_l$, $k=1,\ldots,D_{l-1}$. We assign independent zero-mean Gaussian-process priors
\begin{equation}\label{eq:prior_phi_cov}
\phi^{(l)}_{jk}
\sim
\mathcal{GP}\!\left(0,\tau_l\,\mathcal K(\cdot,\cdot)\right),
\qquad l=1,\ldots,L,
\end{equation}
where \(\tau_l>0\) is a layer-specific prior variance scale. Larger \(\tau_l\) allows more variable link functions in layer \(l\), while smaller \(\tau_l\) imposes stronger shrinkage toward zero. For the observed training inputs, define the vector of link evaluations
$\mathbf z^{(l)}_{jk}
=
\left(
\phi^{(l)}_{jk}(x^{(l-1)}_{1k}),
\ldots,
\phi^{(l)}_{jk}(x^{(l-1)}_{nk})
\right)^\top
\in\mathbb{R}^n$. 
Conditional on the layer-\((l-1)\) outputs, the GP prior implies
\begin{equation}\label{eq:z_conditional_prior}
\mathbf z^{(l)}_{jk}\mid \mathbf X^{(l-1)}
\sim
\mathcal N\!\left(
\mathbf 0,\,
\tau_l \mathbf Q^{(l-1)}_{k}
\right),
\end{equation}
where $\mathbf Q^{(l-1)}_{k}
=
\left[
\mathcal K\!\left(x^{(l-1)}_{ik},x^{(l-1)}_{i'k}\right)
\right]_{i,i'=1}^n$. The vectors \(\mathbf z^{(l)}_{jk}\) are conditionally independent across \((j,k)\), given \(\mathbf X^{(l-1)}\). 
The output of the \(j\)th unit in layer \(l\) is the sum of incoming link evaluations:
$\mathbf X^{(l)}_{\cdot j}
=
\sum_{k=1}^{D_{l-1}}
\mathbf z^{(l)}_{jk}$,  $j=1,\ldots,D_l$.
Therefore, conditional on \(\mathbf X^{(l-1)}\),
\begin{equation}\label{eq:layer_conditional_prior}
\mathbf X^{(l)}_{\cdot j}\mid \mathbf X^{(l-1)}
\sim
\mathcal N\!\left(
\mathbf 0,\,
\tau_l \sum_{k=1}^{D_{l-1}}\mathbf Q^{(l-1)}_{k}
\right).
\end{equation}
This conditional Gaussian representation is the key Bayesian analogue of the forward pass. Marginally across layers, the distribution of \(\mathbf X^{(l)}\) is generally non-Gaussian for \(l\ge 2\), because the covariance matrices in \eqref{eq:layer_conditional_prior} depend on the random outputs from the previous layer.


\begin{proposition}[Marginal moments of hidden neurons]\label{prop:prior_layers}
If the kernel is normalized so that \(\mathcal K(t,t)=1\), then
\begin{equation}\label{eq:prior_mean_var}
\mathbb E\!\left[x^{(l)}_{ij}\right]=0,
\qquad
\mathrm{Var}\!\left[x^{(l)}_{ij}\right]
=
\tau_l D_{l-1}.
\end{equation}
\end{proposition}

\begin{remark}[Variance scaling]
Equation~\eqref{eq:prior_mean_var} shows that the conditional prior variance grows linearly with the number of incoming links \(D_{l-1}\). A variance-preserving choice is therefore \(\tau_l=\tau/D_{l-1}\), which gives
$\mathrm{Var}\!\left[x^{(l)}_{ij}\right]=\tau$.
This scaling is analogous in spirit to variance-preserving initialization rules for neural networks \citep{lecun2002efficient}, but here it is expressed at the function-prior level.
\end{remark}

\subsection{Posterior Distribution and MAP Equivalence}
\label{subsec:posterior}

We now connect the hierarchical GP prior to the penalized estimator in
\eqref{eq:loss}--\eqref{eq:penalty}. Assume the Gaussian regression model
$y_i=\eta_{\text{Wahkon}}(\mathbf x_i)+\epsilon_i$,
 $\epsilon_i\overset{\mathrm{iid}}{\sim}\mathcal N(0,\sigma^2)$.  Under the GP priors in \eqref{eq:prior_phi_cov}, the negative log-posterior for
the collection of link functions is, up to constants,
\begin{equation}\label{eq:function_posterior}
\frac{1}{2\sigma^2}
\sum_{i=1}^{n}
\bigl[y_i-\eta_{\text{\Wahkon}}(\mathbf x_i)\bigr]^2
+
\frac12
\sum_{l=1}^{L}\frac{1}{\tau_l}
\sum_{j=1}^{D_l}\sum_{k=1}^{D_{l-1}}
\bigl\|\phi^{(l)}_{jk}\bigr\|_{\mathcal H_{\mathcal K}}^2 .
\end{equation}
Thus the posterior mode is obtained by balancing data fidelity against the
RKHS norm of each learned link function. Comparing \eqref{eq:function_posterior} with the penalized objective
\eqref{eq:loss}--\eqref{eq:penalty} shows that the two criteria agree after the
variance--penalty calibration
$\tau_l=\frac{\sigma^2}{n\lambda_l},l=1,\ldots,L.$
This yields the following equivalence.

\begin{theorem}[MAP--penalized equivalence]\label{thm:MAP}
Suppose the Gaussian-noise model holds and each link function has the independent
GP prior in \eqref{eq:prior_phi_cov}. If
$
\tau_l=\frac{\sigma^2}{n\lambda_l},
\qquad l=1,\ldots,L,
$
then the maximum a posteriori estimator of the link functions coincides with the
minimizer of the penalized least-squares problem
\eqref{eq:loss}--\eqref{eq:penalty}.
Equivalently, under the representer parameterization of
Theorem~\ref{thm:deep-representer}, if
\[
\mathcal C^*=\{\mathbf C^{(l)*}:l=1,\ldots,L\}
\]
minimizes the finite-dimensional objective \eqref{eq:pls_C}, then the fitted
link-evaluation vectors
\[
\mathbf z^{(l)*}_{jk}
=
\mathbf Q^{(l-1)}_{k}\,
\mathbf c^{(l)*}_{\cdot jk},
\qquad
j=1,\ldots,D_l,\quad k=1,\ldots,D_{l-1},\quad l=1,\ldots,L,
\]
give the posterior mode under the hierarchical GP model.
\end{theorem}

The theorem extends the classical equivalence between smoothing splines and
Gaussian-process MAP estimation \citep{wahba1978improper,wahba1990spline} to the
deep compositional setting of Wahkon. The regularization parameters
\(\lambda_l\) can therefore be interpreted as ratios of observation noise
variance to layer-specific prior variance.

\section{Asymptotic Analysis}\label{sec:thm}


Let $\widehat{\eta}_{\text{Wahkon}}$ be a minimizer of the penalized objective
\eqref{eq:loss}--\eqref{eq:penalty}. For simplicity of exposition, assume that
the layer-specific penalties have a common order,
\[
\lambda_l \asymp \lambda_n,
\qquad l=1,\ldots,L.
\]
Let $\mathcal F_{\text{Wahkon}}$ denote the class of functions represented by
the architecture $(D_0,\ldots,D_L)$ under the RKHS-norm constraints stated below.
We define the oracle approximation
\[
\eta^*
\in
\arg\min_{\tilde \eta\in\mathcal F_{\text{Wahkon}}}
\int \bigl[\tilde\eta(\mathbf x)-\eta(\mathbf x)\bigr]^2\,dP_{\mathbf X},
\]
and the associated approximation error
\[
\Delta_{\mathrm{approx}}^2
=
\int \bigl[\eta^*(\mathbf x)-\eta(\mathbf x)\bigr]^2\,dP_{\mathbf X}.
\]
When the target lies in the Wahkon class, this term is zero. More generally,
it quantifies the bias induced by using a finite architecture and a chosen RKHS.

\begin{theorem}[Convergence rate]\label{thm:rate}
Under the regularity Assumptions~\ref{assump:bounded_domain}--\ref{assump:entropy_rkhs} provided in the Appendix, Section \ref{sec:Pre_proof_thm}, we have
\[
\int
\bigl[
\widehat{\eta}_{\text{Wahkon}}(\mathbf x)-\eta(\mathbf x)
\bigr]^2
\,dP_{\mathbf X}
=
O_p\!\left(
\Delta_{\mathrm{approx}}^2
+
\lambda_n
+
A_D n^{-1}\lambda_n^{-1/(2\alpha)}
\right),
\]
\end{theorem}
where $A_D= \left(\prod_{l=0}^{L}D_l\right)^{
\frac{\alpha+1-\beta}{\alpha}}.$
The three terms in Theorem~\ref{thm:rate} have distinct interpretations.
The first is the approximation error of the chosen architecture. The second is
the bias introduced by RKHS regularization. The third is the stochastic error,
controlled by the entropy of the Wahkon class. Balancing the second and third
terms gives
$\lambda_n
\asymp
\left(\frac{A_D}{n}\right)^{\frac{2\alpha}{2\alpha+1}}
.$
The rate displays the role of the architecture. Increasing width or
depth may reduce $\Delta_{\mathrm{approx}}^2$ by enlarging the approximation
class, but it can also increase the error through $A_D$. The norm
decay parameter $\beta$ governs this tradeoff. When link norms decrease rapidly
with width, the entropy contribution is moderated; when the decay is weak,
additional units can increase complexity and slow convergence. Thus the theorem
does not simply advocate larger networks; it specifies how architectural growth
must be balanced with RKHS regularization to maintain favorable finite-sample
behavior.

\section{Simulation Studies}\label{sec:simulation}

We evaluate Wahkon to (i) verify prior behavior across depth and (ii) assess predictive accuracy on four benchmark functions spanning compositional and nested structures. Full data-generating mechanisms, architectures, and baseline configurations are provided in Appendix~\ref{supp:sim}.

\subsection{Prior behavior across layers}

To examine the hierarchical GP prior, we simulated $1{,}000$ draws of hidden-layer outputs in a depth-$5$ Wahkon with moderate width, using the Gaussian kernel $\mathcal{K}(x,y)=\exp\{-(x-y)^2\}$.
As shown in Figure~\ref{fig:prior}, the first layer is approximately Gaussian, as expected from the prior. Deeper layers exhibit significant non-Gaussianity.
The prior distribution exhibits two notable characteristics (Figure~\ref{fig:prior}b):
\begin{itemize}
  \item[1.] \textbf{Concentration around zero}: Approximately 30\% of Monte Carlo samples from deeper layers ($l \geq 3$) show squared Mahalanobis distances less than 50, suggesting high probability mass around the prior mean $\EE[\mathbf{X}_{\cdot,d_l}^{(l)}] = \mathbf{0}$.
  \item[2.] \textbf{Heavy-tailed behavior}: Approximately 20\% of samples demonstrate extreme distances ($>150$), exceeding the 99.91\% quantile of the $\chi^2_{100}$ reference distribution.
\end{itemize}

This ``spike-and-slab--like'' shape indicates automatic shrinkage of many links while preserving the capacity for a subset to remain active.
This dual behavior may explain Wahkon’s empirical
success in balancing model complexity with generalization performance.

\begin{figure}[!htb]
  \centering
  \includegraphics[width=0.7\textwidth]{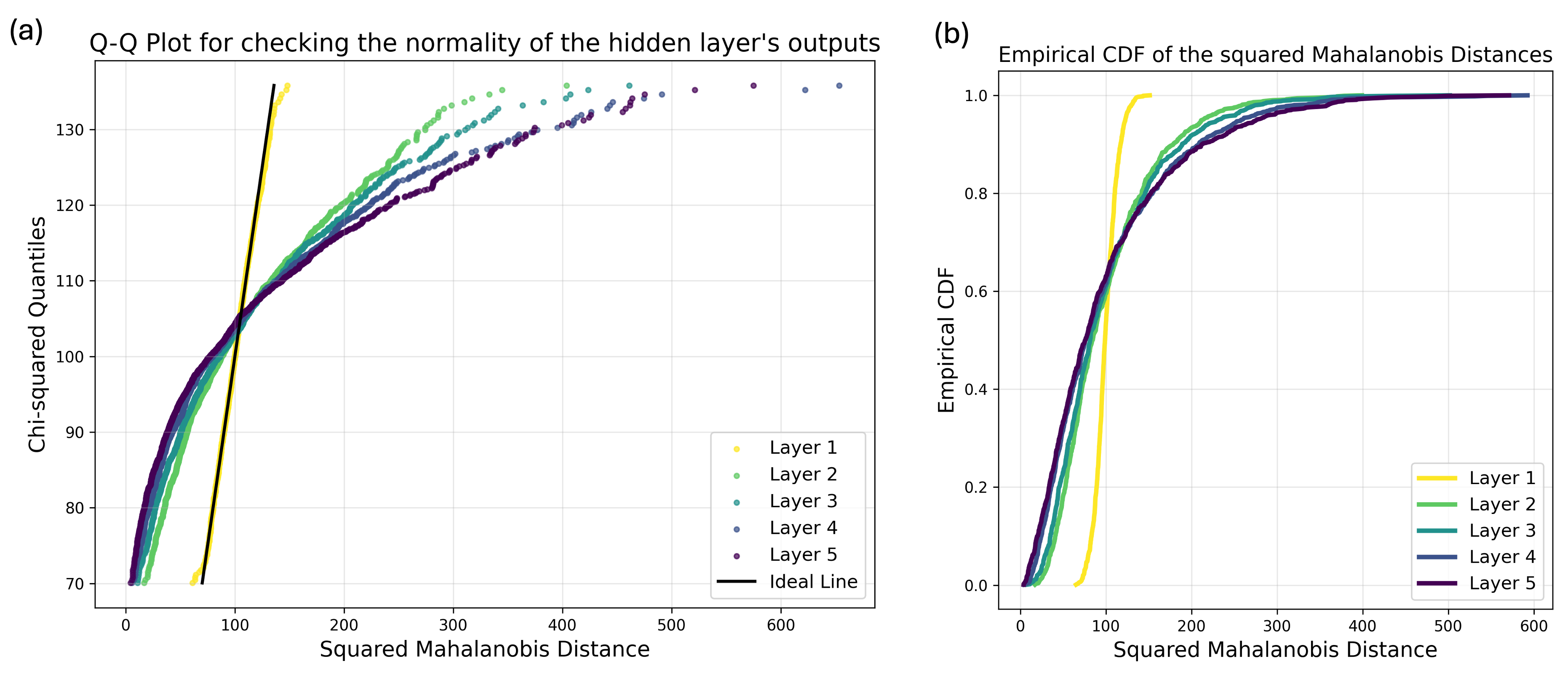}
  \vspace{-10pt}
  \caption{
    Prior distribution analysis for hidden layer outputs in a $[4,4,4,4,4,4]$-Wahkon model.
    (a) Q-Q plot comparing empirical quantiles of squared Mahalanobis distances (x-axis) against theoretical $\chi^2_{100}$ quantiles (y-axis). The black line indicates expected behavior under a 100-dimensional multivariate normal distribution.
    (b) Empirical cumulative distribution function (CDF) of squared Mahalanobis distances.
  }
  \label{fig:prior}
\end{figure}

\subsection{Function-approximation benchmarks}\label{sec:benchmarks}
We compare Wahkon against three competing methods: KAN \citep{liu2024kan}, MLP, and NTK \citep{jacot2018neural}. KAN and Wahkon share the same network width (neuron counts per layer) and grid size $G=9$. The MLP maintains the same depth as the Wahkon/KAN architecture, with each hidden layer's neuron count scaled by $\sqrt{G}=3$ to match the parameter count of a KAN layer with grid size~$G$. The NTK is the infinite-width kernel induced by using the same depth as WKN/KAN in each experiment.  Full implementation details are provided in Appendix~\ref{supp:sim}. 

We evaluate all methods on four benchmark functions ($f_1$--$f_4$) that span highly nonlinear structure ($f_1$, $D=3$), high-dimensional radial oscillation ($f_2$, $D=10$), additive composition ($f_3$, $D=4$), and nested composition ($f_4$, $D=6$). Responses are generated as $y_i = f_k(\bx_i) + 0.1\,\varepsilon_i$ with $\varepsilon_i \sim \mathcal{N}(0,1)$ and inputs from $\mathrm{Unif}[-1,1]^D$. Each experiment is repeated 100 times on randomly generated replicate samples for each setting. The explicit formula of $f_1$--$f_4$ and corresponding model architectures are provided in Appendix~\ref{supp:sim}.

Fig.\ref{fig:simu} plots test $\log(\mathrm{RMSE})$ evaluated on $1,000$ randomly generated testing data points versus training sample size $\in [100,200,400,800,1600,3200]$ for all four functions. We found that Wahkon attains uniformly lower error than KAN, MLP, and NTK.
MLPs underperform despite similar parameter counts, suggesting that learning univariate links within an RKHS is more sample-efficient than employing fixed activations. 
KANs narrow the gap at larger sample sizes, but the absence of explicit RKHS regularization limits their finite-sample behaviors, leading to slow convergence for small samples compared to Wahkon.
NTK performs competitively at small $n$ but plateaus earlier.

\begin{figure}[t]
  \centering
  \includegraphics[width=1\textwidth]{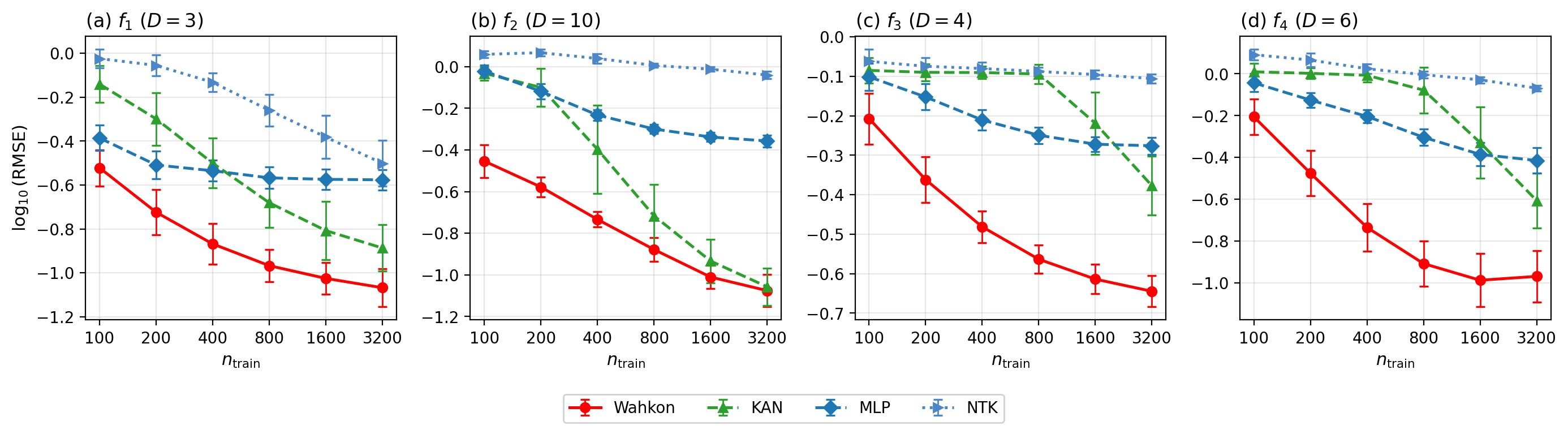}
  \caption{Test RMSE (log scale) versus training sample size for four benchmark functions. Curves show means over $100$ replications; error bars indicate $\pm$ one standard deviation. 
  }
  \label{fig:simu}
\end{figure}

\section{Application: Predicting Surface Proteins from scRNA-seq}\label{sec:cite}

Single-cell multi-omics promises richer biological insight but remains costly and technically challenging at scale. CITE-seq \citep{stoeckius2017large} profiles RNA and surface proteins jointly, yet antibody panels are expensive, subject to batch effects, and can suffer from dropout. Accurately predicting protein abundances from RNA would lower costs for protein panel design, and reveal RNA--protein relationships \citep{liu2024orthogonal,ma2025bisection}. 

We analyzed a human bone marrow mononuclear cell dataset \citep{Stuart2019} with $30{,}672$ cells, $17{,}009$ genes, and $25$ surface proteins. RNA counts were library-size normalized and log-transformed, followed by selection of the top $1{,}000$ most variable genes and PCA reduction to $30$ principal components. $25$ protein's abundance, measured by Antibody-derived tag (ADT) counts were transformed using centered log-ratio (CLR) normalization. The resulting $30$-dimensional RNA PCA matrix served as input features, and each of the $25$ protein targets was predicted independently.

For each protein, we trained a separate predictor using $10 \%$ randomly selected cells for training and the remaining for testing. Wahkon used the architecture $[30\to 15\to 15 \to 15 \to 1]$ with the same training configuration and baseline methods described in Section~\ref{sec:benchmarks}. Each method was fit across $20$ random splits to assess stability, and performance was summarized by test RMSE.
Figure~\ref{fig:real} reports test MSE across the $25$ proteins. Wahkon achieved the lowest average error compared with KAN, MLP, and NTK, and was best for the majority of targets with smaller variability across repetitions. Gains were particularly pronounced for CD45 isoforms (CD45RA, CD45RO), which differentiate na\"ive and memory T cells.

\begin{figure}[t]
  \centering
  \includegraphics[width=0.8\textwidth]{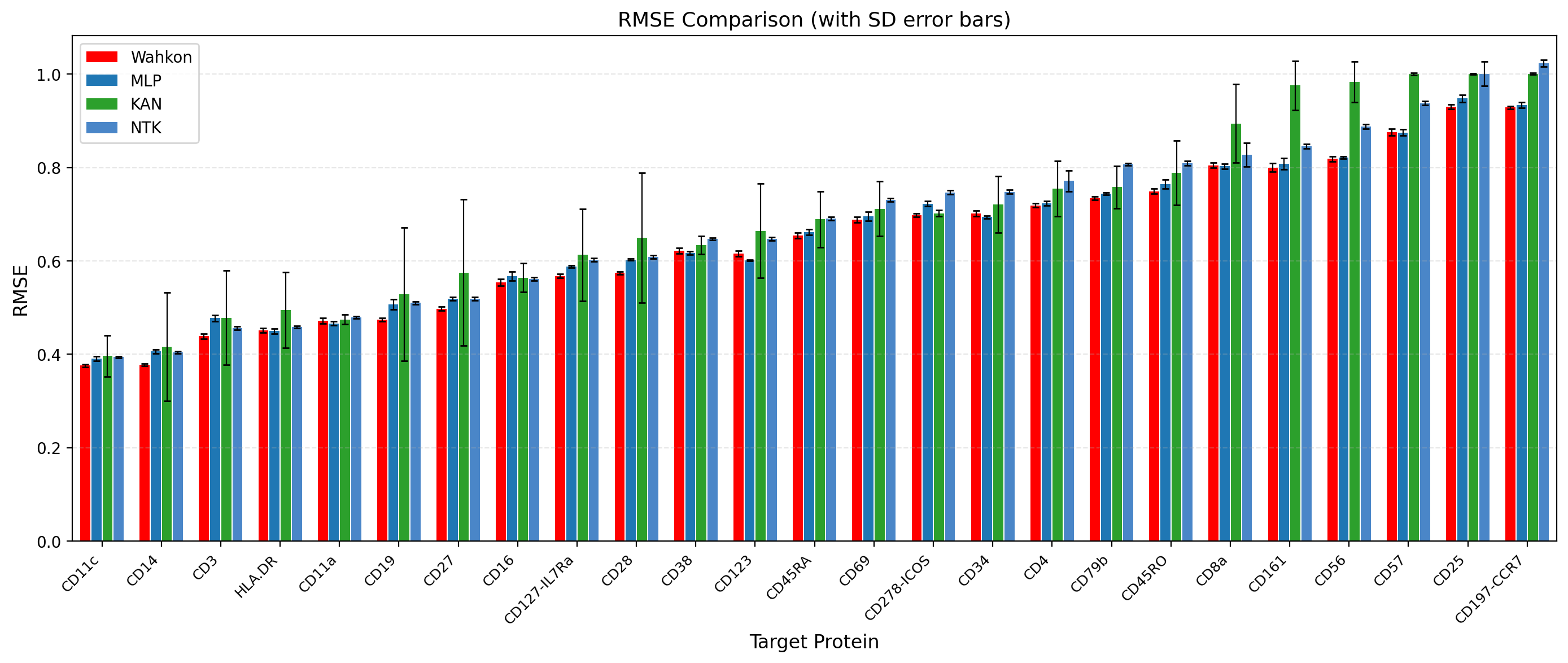}
  \vspace{-10pt}
  \caption{Prediction error for $25$ surface proteins in human bone marrow CITE-seq. Bars show mean test RMSE over repeated random splits; error bars indicate one standard deviation. }
  \label{fig:real}
  \vspace{-10pt}
\end{figure}


\section{Conclusion}\label{sec:conclusion}

Wahkon unifies Kolmogorov's compositional view with RKHS regularization to produce a deep model that is simultaneously expressive, and statistically controlled. The architecture replaces fixed activations by learned univariate links, and the hierarchical penalty enforces layerwise smoothness.
This foundation supports a Bayesian formulation in which the penalized estimator is the MAP under a hierarchical Gaussian-process prior. Entropy bounds and convergence results clarify how depth and width interact with smoothness, and empirical studies align with these predictions while demonstrating gains over strong baselines, including MLP, KAN, and NTK.

There are limits and opportunities. Computational cost grows with sample size and network width, suggesting approximate kernels (Nystr\"om, random features) and low-rank structure as natural accelerations. 
Kernel choice influences performance; data-adaptive or learned kernels and sparsity-promoting priors over links may further improve accuracy and interpretability. 
By providing a principled bridge between modern representation learning and classical statistical inference, Wahkon offers a template for building deep models that are not only accurate but also reliable and explainable.

Another promising direction is uncertainty quantification via the Bayesian equivalence established in Section~\ref{sec:bayes}. 
In practice, full posterior inference over all layers is computationally prohibitive due to the nonlinear coupling between layers. Scalable approximations, such as variational Bayesian inference, offer practical pathways but may underestimate uncertainty. Developing efficient posterior inference methods that propagate uncertainty through the compositional structure remains an important open problem and a natural next step for this framework.

\bibliographystyle{chicago}
\bibliography{ref}

\appendix

\begin{center}
{\large\bf Appendix}
\end{center}
\section{Derivations of Profile Objective}\label{sec:derive_Profile}
By the standard kernel ridge regression identity, the fitted final-layer values are 
$$\widehat{\mathbf X}^{(L)}
=
\mathbf K^{(L-1)}
\bigl(\mathbf K^{(L-1)}+n\lambda_L\mathbf I_n\bigr)^{-1}
\mathbf y,$$
and the minimum value of the final-layer objective is
$$n\lambda_L\,
\mathbf y^\top
\bigl(\mathbf K^{(L-1)}+n\lambda_L\mathbf I_n\bigr)^{-1}
\mathbf y$$.
Therefore, after profiling out the last layer, we obtain the objective
\begin{equation}
\mathcal P(\mathcal C_{<L})
= 
\min_{\{\mathbf c^{(L)}_{\cdot 1k}\}_{k=1}^{D_{L-1}}} \mathcal L(\mathcal C) = 
n\lambda_L\,
\mathbf y^\top
\bigl(\mathbf K^{(L-1)}+n\lambda_L\mathbf I_n\bigr)^{-1}
\mathbf y
+
n\sum_{l=1}^{L-1}\lambda_l
\sum_{j=1}^{D_l}\sum_{k=1}^{D_{l-1}}
\bigl\|\mathbf c^{(l)}_{\cdot jk}\bigr\|^2_{\mathbf Q^{(l-1)}_k}.
\end{equation}
\section{Implementation Details}\label{sec:implementation}

\subsection{Kernel and Efficient computation}
The Gaussian kernel is used throughout:
\[
\mathcal{K}(x,y)=\exp\!\left(-\frac{(x-y)^2}{2\ell^2}\right),\qquad \ell=0.5.
\]
\paragraph{Grid-based approximation.}
The exact representer form uses kernel evaluations at the \(n\) layer-specific
inputs for each link, which can be costly when \(n\) is large. In implementation,
we use a fixed grid of \(G\) inducing points
$u_1<\cdots<u_G$ for each univariate link. A link is approximated by $\phi^{(l)}_{jk}(t)
\approx
\sum_{g=1}^{G} a^{(l)}_{gjk}\,\mathcal K(u_g,t)$,
where \(a^{(l)}_{gjk}\) are the grid coefficients. Let $\mathbf R^{(l-1)}_{k}
=
\left[
\mathcal K(x^{(l-1)}_{ik},u_g)
\right]_{i=1,\ldots,n;\ g=1,\ldots,G}$
be the \(n\times G\) cross-kernel matrix between the layer-\((l-1)\) inputs and
the inducing grid. Then evaluations of the link on the training data are
approximated by \(\mathbf R^{(l-1)}_{k}\mathbf a^{(l)}_{\cdot jk}\), reducing
the per-link evaluation cost from \(O(n^2)\) in the full representer expansion to
\(O(nG)\). The RKHS norm is approximated by
$\bigl\|\phi^{(l)}_{jk}\bigr\|_{\mathcal H_{\mathcal K}}^2
\approx
\mathbf a^{(l)\top}_{\cdot jk}
\mathbf K_{UU}
\mathbf a^{(l)}_{\cdot jk},
\qquad
\mathbf K_{UU}=[\mathcal K(u_g,u_{g'})]_{g,g'=1}^{G}$. 
This grid approximation is used only for computation; the theoretical results
above are stated for the exact representer form. Further implementation details
are provided in Appendix~\ref{supp:sim}.

Each univariate link function is parameterized by a kernel expansion on $G=9$ equally spaced inducing points that remain fixed during training. 
The grid-based parameterization reduces kernel matrix construction from $O(n^2)$ to $O(nG)$ per link function. Mini-batch training further limits the per-step cost: at each iteration, a random subset of $B=200$ observations is drawn, and the kernel matrices and profile objective are evaluated on this subset. 
After the mini-batch optimization converges (or early stopping triggers), the last-layer representer coefficients are refitted using the full training set by solving the closed-form expression in~\eqref{eq:profile}.
Since the profile objective concentrates out the last layer's coefficients, the number of parameters updated by the batch gradient descent is reduced by $n\times D_{L-1}$ relative to joint optimization, yielding faster convergence and lower memory usage.

\subsection{Optimization}
The lower-layer coefficients $\{\mathbf{C}^{(l)}\}_{l=1}^{L-1}$ are optimized by minimizing the profile objective~\eqref{eq:profile} via the Adam optimizer with learning rate $0.005$, mini-batch size $200$, and a maximum of $500$ gradient steps. 

Early stopping is employed Early stopping is used by reserving $20\%$ of the training data as a validation set.
Training stops if the validation loss does not improve by at least $10^{-5}$ for $50$ consecutive steps, and the model snapshot with the lowest validation RMSE is retained.



\subsection{Regularization hyperparameter selection}\label{sec:hyperparams}

The profile objective involves layer-specific penalties $\lambda^{(1)},\dots,\lambda^{(L)}$. In practice, we group these into two parts: a common penalty $\lambda^{(1)} = \cdots = \lambda^{(L-1)} = \lambda_{lower}$ for the feature-learning layers and a separate penalty $\lambda^{(L)}$ for the last layer. These play distinct roles: $\lambda_{\mathrm{lower}}$ controls the smoothness and complexity of the learned feature embedding, while $\lambda^{(l)}$ governs the bias--variance tradeoff in the final regression layer.

The lower-layer penalty is fixed at $\lambda_{\mathrm{lower}} = n_{train}^{-4/5} \times \#\text{links}$, where $\#\text{links} = \sum_{l=1}^{L} D_{l-1} D_l$ is the total number of univariate link functions in the network. This scaling follows from the minimax-optimal rate in Theorem~\ref{thm:rate}. 
In practice, we scale $\lambda_{\mathrm{lower}}$ by the sum of link functions rather than the product of widths, because we tend to overparameterize the network to obtain stable optimization and predictive performance. 
An underlying intuition is that the true target function, when expressed in the Wahkon class, typically requires a high intrinsic dimension at only one or a few layers, while the remaining layers have effectively low complexity. The sum-based scaling is more suitable than a worst-case product bound in practice.
The last-layer penalty $\lambda^{(L)}$ is selected via Bayesian optimization with $5$-fold cross-validation, as described in Section~\ref{sec:hyperparams}.


The last-layer penalty $\lambda^{(l)}$ is selected via Bayesian optimization \citep{snoek2012practical} (BO) with $5$-fold cross-validated RMSE as the evaluation criterion. We use a Mat\'ern-$5/2$ GP surrogate with Expected Improvement acquisition, $15$ total evaluations ($5$ random initial points followed by $10$ BO iterations), and search over $\lambda^{(l)} \in [0.01\,s,\, 3.0\,s]$ where $s = n^{-4/5} \times \#\text{links}$. Each BO evaluation trains the full network from scratch with the candidate $\lambda^{(l)}$ and fixed $\lambda_{\mathrm{lower}}$, ensuring that the lower-layer features are consistent with the regularization being evaluated.

\section{Simulation and Experiment Details}\label{supp:sim}

This section provides the data-generating models, network architectures, and training configurations used in Sections~\ref{sec:simulation} and~\ref{sec:cite}. All methods share the same configuration; the single-cell CITE-seq application differs only in the network widths.

\subsection{Benchmark functions}

Inputs are i.i.d.\ sampled from $\mathrm{Unif}[-1,1]^D$ and responses are $y_i=f_k(\bx_i)+0.1\,\epsilon_i$ with $\epsilon_i\sim\mathcal{N}(0,1)$. The four benchmark functions are:
\begin{align*}
\text{($f_1$) }& f_1(\bx) = \log(x_1^2\!+\!x_2^2\!+\!|\tan x_3|) + \cot\!\bigl(\tfrac{\pi}{1+e^{x_1^2+\sin 6x_2+x_3^2}}\bigr), & D=3\\[3pt]
\text{($f_2$) }& f_2(\bx) = \sin\!\bigl(\textstyle\sum_{i=1}^{10} x_i^2\bigr), & D=10\\[3pt]
\text{($f_3$) }& f_3(\bx) = \exp\!\bigl\{\tfrac12[\sin(\pi(x_1^2\!+\!x_2^2))+\sin(\pi(x_3^2\!+\!x_4^2))]\bigr\}, & D=4\\[3pt]
\text{($f_4$) }& f_4(\bx) = \exp\!\bigl(\sin(\pi(x_1^2\!+\!x_2^2))\bigr)\cos(\pi x_3 x_4), & D=6
\end{align*}
where $f_4$ has two inactive variables ($x_5$, $x_6$ are noise).

\[
\begin{array}{c|c|c}
\toprule
\text{Function} & \text{Input dim.} & \text{Wahkon/KAN width}\\
\midrule
f_1\text{ (nonlinear)} & 3 & [3\to 6 \to 6 \to 1]\\
f_2\text{ (radial osc.)} & 10 & [10\to 10 \to 10 \to 1]\\
f_3\text{ (compositional)} & 4 & [4\to 4 \to 4 \to 1]\\
f_4\text{ (nested comp.)} & 6 & [6\to 6 \to 6 \to 6 \to 1] \\
\bottomrule
\end{array}
\]

\subsection{Baselines configuration}

\paragraph{KAN.} Uses the original implementation of \citet{liu2024kan} with B-spline basis, grid size $G=9$ (matching Wahkon), LBFGS optimizer (learning rate $1.0$ by default), $500$ training steps, sparsity penalty tuned via Bayesian optimization (matching Wahkon), and mini-batch training. The network width (neuron counts) matches Wahkon.

\paragraph{MLP.} A standard feedforward network with ReLU activations. The MLP maintains the same depth (number of hidden layers) as the Wahkon/KAN architecture, with each hidden layer's neuron count scaled by $\sqrt{G}=3$ to approximately match the parameter count of a KAN with grid size~$G=9$. Trained with Adam (learning rate $0.005$), mini-batch size $200$, and early stopping (patience $50$).

\paragraph{NTK.} The infinite-width neural tangent kernel \citep{jacot2018neural} is computed using the \texttt{neural-tangents} library \citep{novak2020neural} with a depth-matched architecture and Erf activation. The noise variance $\sigma^2$ for the NNGP posterior is selected by Bayesian optimization.

\section{Profile vs.\ Direct Objective: Empirical Comparison}\label{app:profile-empirical}

\begin{figure}[t]
  \centering
  \includegraphics[width=0.95\textwidth]{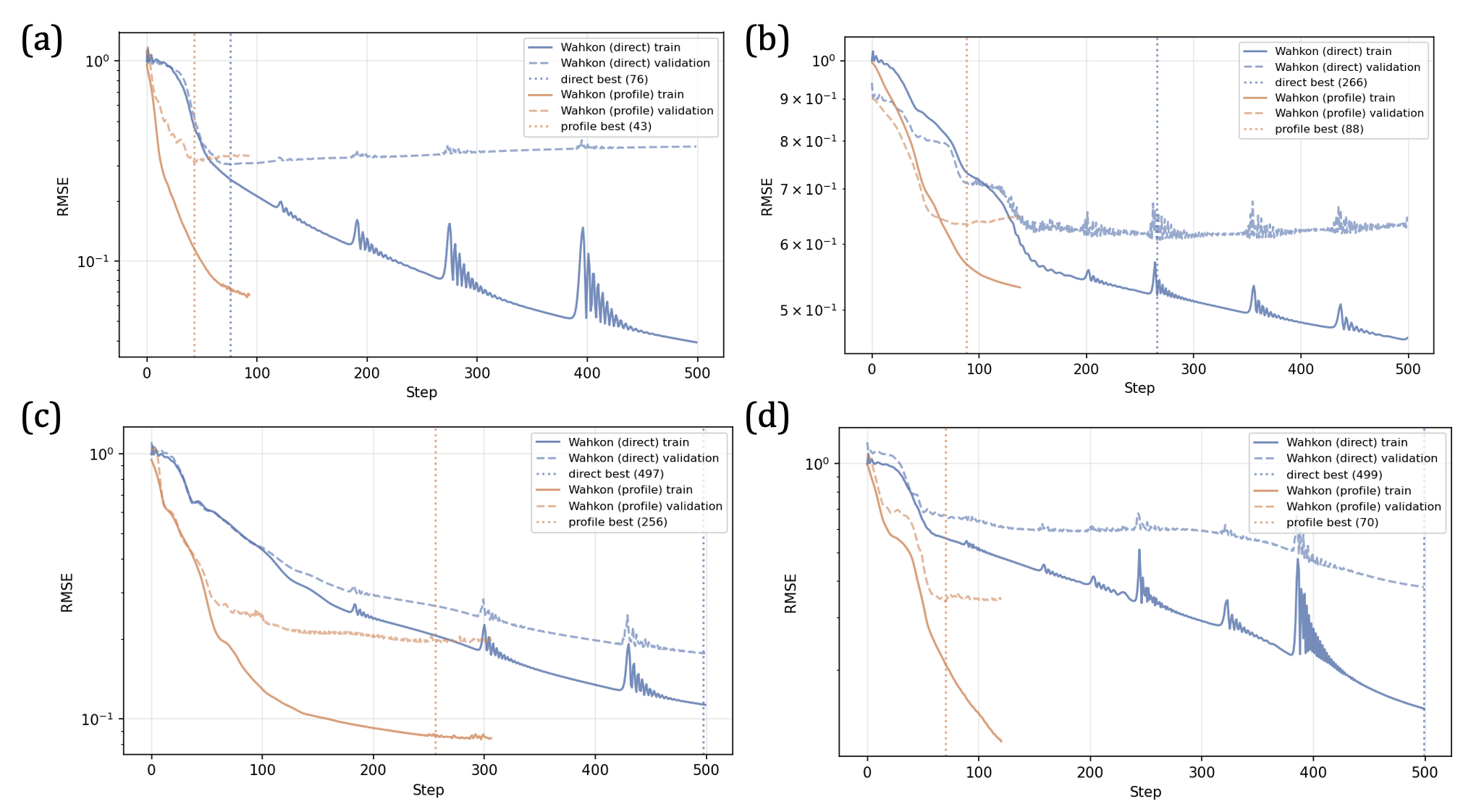}
  \caption{Convergence comparison of the profile objective versus the joint objective on benchmark functions $f_1$--$f_4$. Each panel shows training RMSE (left axis) and test RMSE (right axis) as a function of training steps. The dashed vertical line marks the early-stopping snapshot selected by validation performance. The profile objective consistently reaches lower test RMSE in fewer iterations across all four settings.}
  \label{fig:profile_vs_direct}
\end{figure}

This section presents empirical evidence comparing the profile objective~\eqref{eq:profile} against the direct optimization in \eqref{eq:loss} in which all representer coefficients, including the last layer, are updated by gradient descent. We use the same network architecture and hyperparameters for both objectives on each benchmark function, differing only in how the last-layer coefficients are handled.

As shown in Fig.~\ref{fig:profile_vs_direct}, the profile objective shows faster convergence across all four functions, typically reaching comparable testing loss in $2$--$3\times$ fewer steps.
Notably, the training curves of the direct objective exhibit periodic oscillatory behavior, visible as recurring spikes in both training and test RMSE throughout optimization. We attribute this to the coupling between the last-layer coefficients and the lower-layer kernel matrices in the gradient descent. When updating the lower-layer parameters, the kernel evaluations at the inducing grid points shift, which in turn changes the optimal last-layer solution. Since the direct objective updates both simultaneously by gradient descent, perturbations in the lower layers displace the current last-layer coefficients from their conditional optimum, causing transient increases in loss until the last-layer's coefficients converge. 
The profile objective avoids this pattern entirely as the last-layer coefficients are re-solved analytically at each step.
Furthermore, the profile objective leads to a slightly lower error at the early-stopping point selected by validation performance. 
This improvement may be caused by two factors: (i)~the closed-form last-layer solution eliminates optimization noise in the top layer, providing a cleaner gradient signal to the lower layers; and (ii)~the reduced parameter count (by $n \times D_{L-1}$ representer coefficients) yields a better-conditioned optimization landscape that is less likely to overfit.

\section{Proofs of Main Results}\label{app:proofs}

Throughout, $\mathcal H_{\mathcal K}$ denotes the RKHS with reproducing kernel
$\mathcal K$. For a layer-$l$ link function, we use the notation
$\phi^{(l)}_{jk}$, where $j=1,\ldots,D_l$ indexes the output coordinate and
$k=1,\ldots,D_{l-1}$ indexes the input coordinate from the previous layer.

\subsection{Proof of Theorem~\ref{thm:deep-representer}}
\label{proof:representer}

\begin{proof}
The proof follows the standard RKHS representer argument, applied recursively
through the layers.

For the first layer, fix an edge $(k\to j)$ and consider
$\phi^{(1)}_{jk}\in\mathcal H_{\mathcal K}$. Let
\[
\mathcal S^{(1)}_{k}
=
\mathrm{span}\{\mathcal K(x^{(0)}_{1k},\cdot),\ldots,
\mathcal K(x^{(0)}_{nk},\cdot)\}.
\]
Every $\phi^{(1)}_{jk}$ can be decomposed uniquely as
\[
\phi^{(1)}_{jk}
=
\sum_{i=1}^{n}c^{(1)}_{ijk}\mathcal K(x^{(0)}_{ik},\cdot)
+
\rho^{(1)}_{jk},
\qquad
\rho^{(1)}_{jk}\perp \mathcal S^{(1)}_{k}.
\]
By the reproducing property,
\[
\rho^{(1)}_{jk}(x^{(0)}_{ik})
=
\left\langle
\rho^{(1)}_{jk},\mathcal K(x^{(0)}_{ik},\cdot)
\right\rangle_{\mathcal H_{\mathcal K}}
=0,
\qquad i=1,\ldots,n.
\]
Thus $\rho^{(1)}_{jk}$ does not change any fitted layer-one output on the
training sample, and hence does not change the empirical loss. It only increases
the penalty, since
\[
\|\phi^{(1)}_{jk}\|_{\mathcal H_{\mathcal K}}^2
=
\left\|
\sum_{i=1}^{n}c^{(1)}_{ijk}\mathcal K(x^{(0)}_{ik},\cdot)
\right\|_{\mathcal H_{\mathcal K}}^2
+
\|\rho^{(1)}_{jk}\|_{\mathcal H_{\mathcal K}}^2.
\]
Therefore any minimizer must have $\rho^{(1)}_{jk}=0$.

Now suppose that the representer form holds for all layers before $l$.
Then the layer-$(l-1)$ training outputs
$\{x^{(l-1)}_{ik}:i=1,\ldots,n\}$ are fixed once the lower-layer coefficients
are fixed. Applying the same orthogonal decomposition argument to
$\phi^{(l)}_{jk}$ with respect to
\[
\mathcal S^{(l)}_{k}
=
\mathrm{span}\{\mathcal K(x^{(l-1)}_{1k},\cdot),\ldots,
\mathcal K(x^{(l-1)}_{nk},\cdot)\}
\]
shows that any component orthogonal to $\mathcal S^{(l)}_{k}$ does not change
the fitted values on the training sample and only increases the RKHS penalty.
Hence it must vanish at a minimizer. Therefore,
\[
\phi^{(l)}_{jk}(t)
=
\sum_{i=1}^{n}
c^{(l)}_{ijk}\mathcal K(x^{(l-1)}_{ik},t),
\]
for all $j,k,l$. This completes the induction.
\end{proof}

\subsection{Finite-dimensional form}
\label{proof:finite}

\begin{proof}
By Theorem~\ref{thm:deep-representer},
\[
\phi^{(l)}_{jk}(t)
=
\sum_{i=1}^{n}
c^{(l)}_{ijk}\mathcal K(x^{(l-1)}_{ik},t).
\]
Define the kernel matrix
\[
\mathbf Q^{(l-1)}_{k}
=
\left[
\mathcal K(x^{(l-1)}_{ik},x^{(l-1)}_{i'k})
\right]_{i,i'=1}^{n}
\]
and the coefficient vector
\[
\mathbf c^{(l)}_{\cdot jk}
=
(c^{(l)}_{1jk},\ldots,c^{(l)}_{njk})^\top .
\]
Using the reproducing property,
\[
\begin{aligned}
\|\phi^{(l)}_{jk}\|_{\mathcal H_{\mathcal K}}^2
&=
\left\langle
\sum_{i=1}^{n}c^{(l)}_{ijk}\mathcal K(x^{(l-1)}_{ik},\cdot),
\sum_{i'=1}^{n}c^{(l)}_{i'jk}\mathcal K(x^{(l-1)}_{i'k},\cdot)
\right\rangle_{\mathcal H_{\mathcal K}}  \\
&=
\sum_{i=1}^{n}\sum_{i'=1}^{n}
c^{(l)}_{ijk}c^{(l)}_{i'jk}
\mathcal K(x^{(l-1)}_{ik},x^{(l-1)}_{i'k})  \\
&=
\mathbf c^{(l)\top}_{\cdot jk}
\mathbf Q^{(l-1)}_{k}
\mathbf c^{(l)}_{\cdot jk}.
\end{aligned}
\]
Substituting this identity into the RKHS penalty gives the finite-dimensional
criterion in \eqref{eq:pls_C}. The penalty is quadratic in the coefficients for
fixed layer inputs, while the empirical loss remains generally nonlinear because
the layer outputs are obtained recursively through the network.
\end{proof}

\subsection{Proof of Proposition~\ref{prop:prior_layers}}
\label{proof:prior}

\begin{proof}
Under the prior
\[
\phi^{(l)}_{jk}\sim \mathcal{GP}(0,\tau_l\mathcal K),
\]
and conditional on the layer-$(l-1)$ outputs, the vector of link evaluations
\[
\mathbf z^{(l)}_{jk}
=
\left(
\phi^{(l)}_{jk}(x^{(l-1)}_{1k}),\ldots,
\phi^{(l)}_{jk}(x^{(l-1)}_{nk})
\right)^\top
\]
is multivariate normal:
\[
\mathbf z^{(l)}_{jk}\mid \mathbf X^{(l-1)}
\sim
\mathcal N\left(
\mathbf 0,\tau_l\mathbf Q^{(l-1)}_k
\right).
\]
The vectors $\mathbf z^{(l)}_{jk}$ are conditionally independent across $k$ for
fixed $j$. Since
\[
\mathbf X^{(l)}_{\cdot j}
=
\sum_{k=1}^{D_{l-1}}\mathbf z^{(l)}_{jk},
\]
the sum of independent Gaussian vectors is Gaussian, yielding
\[
\mathbf X^{(l)}_{\cdot j}\mid \mathbf X^{(l-1)}
\sim
\mathcal N\left(
\mathbf 0,\tau_l\sum_{k=1}^{D_{l-1}}\mathbf Q^{(l-1)}_k
\right).
\]
If
$\mathcal K(t,t)=1$, then the $i$th diagonal entry of
$\sum_{k=1}^{D_{l-1}}\mathbf Q^{(l-1)}_k$ equals $D_{l-1}$, and hence
\[
\Var(x^{(l)}_{ij}\mid \mathbf X^{(l-1)})
=
\tau_lD_{l-1},
\qquad
E(x^{(l)}_{ij}\mid \mathbf X^{(l-1)})
=
0.
\]
Therefore, $\mathbb E\!\left[x^{(l)}_{ij}\right]=E\!\left[ E(x^{(l)}_{ij}\mid \mathbf X^{(l-1)}) \right]=0
$,
$\mathrm{Var}\left[x^{(l)}_{ij}\right]=
\mathrm{Var}\!\left[E(x^{(l)}_{ij}\mid \mathbf X^{(l-1)})\right]
+ E\!\left[\mathrm{Var}(x^{(l)}_{ij}\mid \mathbf X^{(l-1)})\right]  =
\tau_l D_{l-1}.$
\end{proof}

\subsection{Proof of Theorem~\ref{thm:MAP}}
\label{proof:map}

\begin{proof}
Assume
\[
y_i=\eta_{\text{Wahkon}}(\mathbf x_i)+\epsilon_i,
\qquad
\epsilon_i\overset{\mathrm{iid}}{\sim}\mathcal N(0,\sigma^2).
\]
The likelihood is proportional to
\[
\exp\left\{
-\frac{1}{2\sigma^2}
\sum_{i=1}^{n}
\bigl[y_i-\eta_{\text{Wahkon}}(\mathbf x_i)\bigr]^2
\right\}.
\]
The independent GP priors on the link functions contribute the function-space
penalty
\[
\exp\left\{
-\frac12
\sum_{l=1}^{L}\frac{1}{\tau_l}
\sum_{j=1}^{D_l}\sum_{k=1}^{D_{l-1}}
\|\phi^{(l)}_{jk}\|_{\mathcal H_{\mathcal K}}^2
\right\}.
\]
Therefore, up to an additive constant, the negative log-posterior is
\[
\frac{1}{2\sigma^2}
\sum_{i=1}^{n}
\bigl[y_i-\eta_{\text{Wahkon}}(\mathbf x_i)\bigr]^2
+
\frac12
\sum_{l=1}^{L}\frac{1}{\tau_l}
\sum_{j=1}^{D_l}\sum_{k=1}^{D_{l-1}}
\|\phi^{(l)}_{jk}\|_{\mathcal H_{\mathcal K}}^2 .
\]
Multiplying by $2\sigma^2$ gives
\[
\sum_{i=1}^{n}
\bigl[y_i-\eta_{\text{Wahkon}}(\mathbf x_i)\bigr]^2
+
\sum_{l=1}^{L}
\frac{\sigma^2}{\tau_l}
\sum_{j=1}^{D_l}\sum_{k=1}^{D_{l-1}}
\|\phi^{(l)}_{jk}\|_{\mathcal H_{\mathcal K}}^2 .
\]
Thus this objective is identical to
\eqref{eq:loss}--\eqref{eq:penalty} when
\[
\frac{\sigma^2}{\tau_l}=n\lambda_l,
\qquad\text{equivalently}\qquad
\tau_l=\frac{\sigma^2}{n\lambda_l}.
\]
The MAP estimator therefore coincides with the penalized least-squares estimator.
The finite-dimensional coefficient form follows from
Theorem~\ref{thm:deep-representer}.
\end{proof}

\subsection{Preliminary Material for the Proof of  Theorem \ref{thm:rate}}\label{sec:Pre_proof_thm}
\paragraph*{Assumptions.}
We first list the assumptions for deriving the minimax-convergence rate.
\begin{assumption}[Bounded domain]\label{assump:bounded_domain}
The predictor domain $\Omega_X\subset\mathbb R^D$ is bounded; that is,
$\|\mathbf x\|_2\le C_X$ for all $\mathbf x\in\Omega_X$.
\end{assumption}

\begin{assumption}[Smooth link functions]\label{assump:smooth_links}
For constants $C_\phi>0$ and $\beta>0$, each link function satisfies
\[
\|\phi^{(l)}_{jk}\|_{\mathcal H_{\mathcal K}}
\le
C_\phi D_l^{-\beta},
\qquad
j=1,\ldots,D_l,\quad k=1,\ldots,D_{l-1}.
\]
The link functions are also uniformly Lipschitz on the relevant compact domain.
\end{assumption}

\begin{assumption}[RKHS entropy]\label{assump:entropy_rkhs}
For each $C>0$, the RKHS ball
$\{\phi:\|\phi\|_{\mathcal H_{\mathcal K}}\le C\}$ satisfies
\[
\log N\!\left(
\delta,
\{\phi:\|\phi\|_{\mathcal H_{\mathcal K}}\le C\},
\|\cdot\|_\infty
\right)
\le
C_2 \delta^{-1/\alpha},
\qquad \delta>0,
\]
for constants $C_2>0$ and $\alpha>0$.
\end{assumption}
The regularity conditions outlined in our analysis are mild and broadly applicable. 
Assumption \ref{assump:bounded_domain}, which requires predictors to reside within a bounded domain, is a standard assumption in both nonparametric regression and neural network theory \citep{gu:13,horowitz2007rate,bauer2019deep}.

Assumption \ref{assump:smooth_links} imposes constraints on the Lipschitz continuity and RKHS norm of the activation functions to balance network expressivity with complexity control.
 Specifically, the RKHS norm bound \( \|\phi^{(l)}_{d_l,d_{l-1}}\|_{\mathcal{H}_\mathcal{K}} \leq C_\phi D_l^{-\beta} \)  enforces a decay rate on the norm of activation functions with respect to the number of neurons in the next layer. 

Assumption \ref{assump:entropy_rkhs} imposes a constraint on the complexity of the function class with the finite norm in RKHS $\cH_{\cK}$. 
The value of $\alpha $is determined by the kernel $\cK$, and a larger $\alpha$ corresponds to a more constrained function class. For example, for a Sobolev space, $\alpha$ is its order \citep{edmunds1996function}.

Under these assumptions, the size of the Wahkon function class can be controlled
by an entropy bound. Let
\[
A_D
=
\left(\prod_{l=0}^{L}D_l\right)^{1+\frac{1}{\alpha}-\frac{\beta}{\alpha}} .
\]
The factor $A_D$ summarizes the effect of the architecture on statistical
complexity. Width and depth increase the number of links, while the decay
$D_l^{-\beta}$ in the RKHS norm moderates this growth.

\begin{theorem}[Entropy bound]\label{thm:entropy}
Under Assumptions~\ref{assump:bounded_domain}--\ref{assump:entropy_rkhs},
\[
\log N\!\left(
\delta,
\mathcal F_{\text{Wahkon}},
\|\cdot\|_2
\right)
\le
C A_D \delta^{-1/\alpha},
\]
where $C>0$ is independent of $n$ and $\delta$.
\end{theorem}

The entropy bound leads to the following convergence result.
\begin{proof}[Proof of Theorem~\ref{thm:entropy}]
Let
\[
A_D
=
\left(\prod_{l=0}^{L}D_l\right)^{
1+\frac{1}{\alpha}-\frac{\beta}{\alpha}}.
\]
We show that the covering number of the Wahkon class is bounded by
$C A_D\delta^{-1/\alpha}$.

For a single univariate link at layer $l$, Assumption~\ref{assump:entropy_rkhs}
and Assumption~\ref{assump:smooth_links} imply
\[
\log N\left(
\epsilon,
\{\phi:\|\phi\|_{\mathcal H_{\mathcal K}}\le C_\phi D_l^{-\beta}\},
\|\cdot\|_\infty
\right)
\le
C D_l^{-\beta/\alpha}\epsilon^{-1/\alpha}.
\]
Now consider one output coordinate of layer $l$,
\[
x^{(l)}_j
=
\sum_{k=1}^{D_{l-1}}
\phi^{(l)}_{jk}(x^{(l-1)}_k).
\]
A $\delta$-cover for the sum can be constructed by using
$\delta/D_{l-1}$-covers for each incoming link. Hence, for one output coordinate,
the log-covering number is bounded by
\[
C D_{l-1}
D_l^{-\beta/\alpha}
\left(\frac{\delta}{D_{l-1}}\right)^{-1/\alpha}
=
C D_l^{-\beta/\alpha}
D_{l-1}^{1+1/\alpha}
\delta^{-1/\alpha}.
\]
Since layer $l$ has $D_l$ output coordinates, the log-covering number for the
layer map is bounded by
\[
C D_l^{1-\beta/\alpha}
D_{l-1}^{1+1/\alpha}
\delta^{-1/\alpha}.
\]

The full Wahkon class is obtained by composing these layer maps. By the
uniform Lipschitz condition in Assumption~\ref{assump:smooth_links}, perturbations
at an intermediate layer propagate to the output by at most a constant factor
depending on the Lipschitz constants and the fixed depth $L$. These constants can
be absorbed into $C$. Combining the layer-wise covers over
$l=1,\ldots,L$ gives
\[
\log N\left(
\delta,
\mathcal F_{\text{Wahkon}},
\|\cdot\|_2
\right)
\le
C
\left(\prod_{l=0}^{L}D_l\right)^{
1+\frac{1}{\alpha}-\frac{\beta}{\alpha}}
\delta^{-1/\alpha}.
\]
This proves the stated entropy bound.
\end{proof}

\subsection{Proof of Theorem~\ref{thm:rate}}
\label{proof:rate}

\begin{proof}
Let $\eta^*$ be the best approximation to $\eta$ in
$\mathcal F_{\text{Wahkon}}$, and define
\[
\Delta_{\mathrm{approx}}^2
=
\|\eta^*-\eta\|_{L^2(P_{\mathbf X})}^2.
\]
By the triangle inequality,
\[
\|\widehat\eta_{\text{Wahkon}}-\eta\|_{L^2(P_{\mathbf X})}^2
\le
2\|\widehat\eta_{\text{Wahkon}}-\eta^*\|_{L^2(P_{\mathbf X})}^2
+
2\Delta_{\mathrm{approx}}^2.
\]
It remains to control the estimation error
$\|\widehat\eta_{\text{Wahkon}}-\eta^*\|_{L^2(P_{\mathbf X})}^2$.

The penalized least-squares basic inequality gives
\[
\frac1n
\sum_{i=1}^{n}
\bigl[
\widehat\eta_{\text{Wahkon}}(\mathbf x_i)-\eta^*(\mathbf x_i)
\bigr]^2
+
\lambda_n J(\widehat\eta_{\text{Wahkon}})
\le
\text{empirical process term}
+
\lambda_n J(\eta^*) .
\]
Under the assumed boundedness and entropy conditions, standard penalized
empirical-process theory for least-squares estimators
\citep[see, e.g.,][]{geer2000empirical} implies
\[
\|\widehat\eta_{\text{Wahkon}}-\eta^*\|_{L^2(P_{\mathbf X})}^2
=
O_p\left(
\lambda_n J(\eta^*)
+
A_D n^{-1}\lambda_n^{-1/(2\alpha)}
\right),
\]
where
\[
A_D
=
\left(\prod_{l=0}^{L}D_l\right)^{
1+\frac{1}{\alpha}-\frac{\beta}{\alpha}} .
\]
If the oracle $\eta^*$ has bounded penalty, $J(\eta^*)=O(1)$, the first term is
$O(\lambda_n)$. Combining this bound with the approximation term yields
\[
\|\widehat\eta_{\text{Wahkon}}-\eta\|_{L^2(P_{\mathbf X})}^2
=
O_p\left(
\Delta_{\mathrm{approx}}^2
+
\lambda_n
+
A_D n^{-1}\lambda_n^{-1/(2\alpha)}
\right),
\]
which proves the theorem.

Balancing the last two terms gives
\[
\lambda_n
\asymp
\left(\frac{A_D}{n}\right)^{\frac{2\alpha}{2\alpha+1}},
\]
and substituting this choice gives the stated optimized rate.
\end{proof}

\end{document}